\documentclass[12pt]{amsart}
\usepackage{latexsym}
\usepackage[dvips]{graphicx}
\usepackage{amsmath,amssymb, amsthm}
\usepackage{graphicx, color}

\newtheorem{thm}{Theorem}

\theoremstyle{definition}
\newtheorem{defi}[thm]{Definition}
\newtheorem{rem}[thm]{Remark}

\newtheorem{ex}{Example}

\title{The Fourier estimation method with positive semi-definite estimators
}
\author{Jir\^o Akahori}
\author{Nien-Lin Liu}
\author{Maria Elvira Mancino}
\author{Yukie Yasuda}
\address[Jir\^o Akahori, Nien-Lin Liu and Yukie Yasuda]{
Department of Mathematical Sciences, Ritsumeikan University
1-1-1 Nojihigashi, Kusatsu, Shiga, 525-8577, Japan}
\address[Maria Elvira Mancino]{Department of Economics and Management,
University of Firenze, Via delle Pandette, 9, Firenze, Italy}
\date{October, 2014}
\numberwithin{equation}{section}
\begin{document}
\maketitle
\thispagestyle{empty}
\begin{abstract}
In this paper we present a slight modification of
the Fourier estimation method of the spot volatility
(matrix) process of a continuous It\^o semimartingale
where the estimators are always non-negative definite.
Since the estimators are factorized,
computational cost will be saved a lot.
\end{abstract}


\section{Introduction}
In this paper we present a slight modification of
the Fourier estimation method of the spot volatility
(matrix) process of an It\^o semimartingale.
The method is originally introduced by Paul Malliavin
and the third author in \cite{MM02,MM09}.
The main aim of the modification is to construct
an estimator of the matrix
which always stays non-negative definite.

A motivation of the present study is 
to make an implementation of 
the Fourier method easier when it is applied
to ``dynamic principal component analysis", an important
application of the spot volatility estimation
(see \cite{LM12,LN12}).
Due to the lack of symmetry of the matrices,
its estimated eigenvalues are sometimes non-positive,
or even worse, non-real.
This is not the case with those based on
finite differences (FD) of the integrated volatility
such as Ngo-Ogawa's method \cite{NO09}.
However, as the Fourier method has many advantages
over the FD ones, among which
robustness against non-synchronous observations
counts most,
the modification to be presented in this paper
would be important.

There is a by-product of the modification;
thanks to a symmetry imposed to have the non-negativity
our estimator is factorized,
which may save computational cost a lot.

\

The present paper is organized as follows.
We will firstly introduce a generic form of
Fourier type estimators (Definition \ref{Generic}),
and discuss how it works as a recall (section \ref{Heurietic}).
After remarking that the classical one is obtained by
a choice of the ``fiber" (Remark \ref{classical}),
we will
introduce a class of such estimators (section \ref{PSD}),
each of which is labeled by a positive definite function.
As a main result, we will prove its positive semi-definiteness
(Theorem \ref{PFE}).
In addition, we give a remark (Remark \ref{reduction}) that
with an action of finite group, some of the newly introduced
positive semi-definite estimators reduce to
the classical one.

In section \ref{Parametrize},
we will give
a factorized representation of the estimator
(Definition \ref{factorize})
which is parameterized by
a measure by way of Bochner's correspondence.
The use of the expression may reduce the computation
cost, as will be exemplified by simple experiments
presented in section \ref{experiments}.
Section \ref{Measure} gives an important remark
that as a sequence of estimators,
the parameterization measures should be
a delta-approximating kernel.
Three examples of the kernels are given
(Examples \ref{Cauchy}--\ref{doubleF}), two of which are used
in the simple experiments presented in section \ref{experiments}.

In the present paper, we will not study limit theorems;
consistency nor central limit theorem in detail.
More detailed studies in these respects will appear in another paper.

\subsection*{Acknowledgment}
This work was partially supported by JSPS KAKENHI Grant Numbers
25780213, 
23330109, 24340022, 23654056 and 25285102.

\section{The Fourier Method Revisited}
\subsection{Generic Fourier Estimator}
Let $ X = (X^1,\cdots, X^d) $ be a $ d $-dimensional
continuous semimartingale.
Suppose that its quadratic variation (matrix) process
is absolutely continuous in $ t $ almost surely.
In this paper, we are interested in a statistical estimation of
the so-called {\em spot volatility} process;
\begin{equation*}
\frac{d[ X^j, X^{j'} ]_t}{dt} (\omega) =: V^{jj'}_t (\omega),
\quad 0 \leq t \leq 1, \ 1\leq j , j' \leq d,
\end{equation*}
as a function in $ t $,
especially when $ d \geq 2 $,
for a given observations of $ X^{j}$
on a partition $\Delta^j : 0 = t^j_0 < \cdots < t^j_{N_j} = 1 $, $ j= 1, \cdots, d$.

Here and hereafter
we normalize the time interval to $ [0,1] $ for notational simplicity.

We start with a generic form of a {\em Fourier estimator}
with respect to this observation set,
to have a unified view.
\begin{defi}\label{Generic}
Let $ \mathcal{K} $ be a finite subset of $ \mathbf{Z}$,
$ \mathcal{S} =\{ \mathcal{S} (k) \subset_{\mathrm{finite}} \mathbf{Z}^2 : k \in{\mathcal K}, 
(s,s') \in \mathcal{S}(k) \Rightarrow s+s'=k \} $
be a ``fiber" on $ \mathcal{K} $, and $ c $
be a complex function on $ \mathcal{K} $.
A Fourier estimator associated with $ (\mathcal{K}, \mathcal{S}, c) $
is a $ d \times d $ matrix defined entry-wisely by
\begin{equation}\label{GF}
\begin{split}
& (V_{(\mathcal{K}, \mathcal{S}, c)})_{j,j'} (t)\\
&= \sum_{l=1}^{N_j} \sum_{l'=1}^{N_{j'}}
\sum_{k \in \mathcal{K}} c(k) e^{2 \pi i k t}
\sum_{(s, s') \in \mathcal{S} (k) }   e^{-2 \pi i s t_l^j}
e^{-2 \pi i s' t_{l'}^{j'}}
\Delta X^j_l\Delta X^{j'}_{l'},\\
& \hspace{4cm} ( 1 \leq j,j' \leq d),
\end{split}
\end{equation}
where
\begin{equation*}
\Delta X^j_l := X^j_{t^j_l} - X^{j}_{t^j_{l-1}}.
\end{equation*}
\end{defi}

\begin{rem}\label{classical}
If we take $ \mathcal{K} = \{ 0,\pm 1, \cdots, \pm L \} $ for some positive
integer $ L $,
$ \mathcal{S}(k) = \{ (s,s') | s+s'=k, s=0, \pm 1, \cdots, \pm M \} $
for some positive integer $ M $, and
$ c (k) = ( 1  - \frac{|k|}{L+1} )/(2M+1) $, then
the estimator $ V_{(\mathcal{K}, \mathcal{S}, c)} $
coincides with the one introduced in
\cite{MM09}. In fact, with these parameters, we have
\begin{equation}\label{FEMM09}
\begin{split}
& (V)_{j,j'} (t) = \sum_{k=-L}^{L}
\left( 1 - \frac{|k|}{L+1} \right) e^{2\pi i k t}\\
& \cdot \frac{1}{2M+1} \sum_{s= -M}^M
\sum_{l=1}^{N_j} \sum_{l'=1}^{N_{j'}} e^{-2 \pi i s t^j_l}
e^{-2 \pi i (k-s) t^{j'}_{l'}}
\Delta X^{j}_l \Delta X^{j'}_{l'}.
\end{split}
\end{equation}
With the Dirichlet and the Fej\'er kernels;
\begin{equation*}
D_{M} (x) =
\sum_{|s|\leq M} e^{2 \pi i s x}
= \frac{\sin (2M+1) \pi x}{\sin \pi x},
\end{equation*}
{and}
\begin{equation*}
\begin{split}
K_{L} (x) &:= \frac{1}{L} \sum_{M=0}^{L-1} D_{M} (x)
=\sum_{|k| \leq L-1}
\left( 1 - \frac{|k|}{L} \right) e^{2\pi i k x}\\
&= \frac{1}{L} \left(
\frac{\sin(L\pi x) }{\sin (\pi x)}
\right)^2,
\end{split}
\end{equation*}
we can rewrite (\ref{FEMM09}) as
\begin{equation}\label{FEMM14}
\begin{split}
& (V)_{j,j'} (t)\\
&=\frac{1}{(2M+1)} \sum_{l,l'} K_{L+1} (t-t^j_l)
D_{M} (t^j_l-t^{j'}_{l'}) \Delta X^j_l\Delta X^{j'}_{l'}\\
&=\frac{1}{(2M+1)}
\sum_{l,l'} \left(
\frac{\sin\{(L+1) \pi (t-t^j_l)\}) }{\sin (\pi (t-t^j_l))}
\right)^2
\frac{\sin \{ (2M+1) \pi (t^j_l -t^{j'}_{l'})\}}{\sin \pi (t^j_l -t^{j'}_{l'})}
\\
& \hspace{4cm} \cdot \Delta X^j_l\Delta X^{j'}_{l'}.
\end{split}
\end{equation}
\end{rem}

\subsection{A Heuristic Derivation}\label{Heurietic}
Here we give a heuristic explanation of
the idea behind the Fourier method, which was
originally proposed in \cite{MM02, MM09}, and now
is extended to (\ref{GF}) to include a class 
of positive semi-definite estimators 
that will be introduced in the next section.

Looking at (\ref{GF}) or (\ref{FEMM09}) carefully,
we notice that though naively, we may suppose
\begin{equation}\label{heuristic}
\begin{split}
&(V_{(\mathcal{K}, \mathcal{S}, c)})_{j,j'} (t) \\
& \sim \sum_{k \in \mathcal{K}} c (k) e^{2 \pi i k t}
\sum_{ (s,s') \in S (k)} \left(\int_0^1 e^{ - 2 \pi i s u}
d X^j_u \right) \left(\int_0^1 e^{ - 2 \pi i s' u}
d X^{j'}_u \right) \\
&= \sum_{k \in K} c (k) |\mathcal{S}(k)| e^{2 \pi i k t}
\int_0^1 e^{ - 2 \pi i k u} d [X^j, X^{j'}]_u  \\
&+ \sum_{k \in K} c (k)  e^{2 \pi i k t}
\int_0^1 \int_0^u \sum_{ (s,s') \in \mathcal{S} (k)}
( e^{ - 2 \pi i s u} e^{- 2 \pi i s' v} dX^j_u d X^{j'}_v \\
& \hspace{6cm}
+ e^{ - 2 \pi i s' u} e^{- 2 \pi i s v} dX^{j'}_u d X^{j}_v ) \\
&=: I + II.
\end{split}
\end{equation}
The term $ I $ can be understood to be a weighted
partial sum of Fourier series of $ V^{jj'} $, which may
converge uniformly if the weight
$ c (k) | \mathcal{S} (k) | $ is properly chosen (in the case of
(\ref{FEMM09}), it is Fej\'er kernel).
The term $ II $ vanishes, roughly because, 
 $ \sum_{ (s,s') \in S (k)} 
( e^{ - 2 \pi i s u} e^{- 2 \pi i s' v} $ behaves like a Dirichlet kernel
$ D (u-v) $, which converges weakly to the delta measure. 


\section{Positive Semi-Definite Fourier Estimators}\label{PSD}

In financial applications, we are often interested in computing
the rank of the (spot) volatility matrix.
Since it is positive semi-definite
the rank is estimated by a hypothesis testing
on the number of positive eigenvalues.
The estimator (\ref{FEMM09}), or equivalently given as (\ref{FEMM14}),
however, sometimes
fails to be symmetric\footnote{This is seen from
the following simple observation:
$ \sum_{l=1}^2\sum_{l'=1}^2 a_{l,l'}
(x^{1}_l x^2_{l'} -x^{2}_l x^1_{l'})
=(a_{1,2} - a_{2,1}) (x_1^1x_2^2-x_1^2 x_2^1) $. }
since $ K_M (t-t_l) D_L (t_l - t_{l'}) $
is not symmetric in $ l,l' $,
and thus its eigenvalues are not always real numbers.
This causes
some trouble in estimating the rank of the volatility matrix.
Here we propose a class of Fourier estimators that will be proven
to be symmetric and positive semi-definite.

\begin{rem}
We just note that the positive semi-definiteness of
the Fourier estimator defined in \cite{MM09}
of the integrated volatility matrix
(the $ 0 $-th coefficient) is proved in \cite{MS11}.
\end{rem}

\subsection{Positive Fourier Estimators}

The main result of the present paper is the following
\begin{thm}\label{PFE}
Suppose that
$ \mathcal{K} = \{ 0, \pm 1, \cdots, \pm 2M \} $
for some positive integer $ M $,
$ c $ is a positive semi-definite function on $ \mathcal{K} $,
and
\begin{equation*}
\begin{split}
& \mathcal{S} (k) = \\
&
\begin{cases}
\{ (-M+k+v, M-v) : v=0,\cdots,2M-k \} & 0 \leq k \leq 2M \\
\{ (M+k-v, -M+v :v=0,\cdots,2M+k\} & -2M \leq k < 0.
\end{cases}
\end{split}
\end{equation*}
Then, $ V_{(\mathcal{K}, \mathcal{S}, c)} $ defined in (\ref{GF}) is positive semi-definite.
\end{thm}

\begin{proof}
Let $ a_j $, $ j=1,2,3 $ be arbitrary functions on $ \mathbf{Z} $.
From the definitions of $ (\mathcal{K},\mathcal{S}) $, we notice that
\begin{equation*}
\begin{split}
& \sum_{ k \in \mathcal{K} } \sum_{ (s,s') \in \mathcal{S}(k) } a_1 (k) a_2 (s) a_3 (s') \\
&= \sum_{ k =0 }^{2M} \sum_{ v=0  }^{2M-k} a_1 (k) a_2 (-M+k+v) a_3 (M-v) \\
&+ \sum_{ k = -2M }^{-1} \sum_{ v=0  }^{2M+k}
a_1 (k) a_2 (M+k-v) a_3 (-M+v) =: A+B
\end{split}
\end{equation*}
For the first term of the right-hand-side,
\begin{equation*}
\begin{split}
A &= \sum_{ k =0 }^{2M} \sum_{ u= k-M }^{M} a_1 (k) a_2 (k-u) a_3 (u) \\
&= \sum_{ u = -M }^{M} \sum_{k=0}^{u+M} a_1 (k) a_2 (k-u) a_3 (u) \\
&= \sum_{ u = -M }^{M} \sum_{u'=-u}^{M} a_1 (u+u') a_2 (u') a_3 (u),
\end{split}
\end{equation*}
where we set $ u = M-v $ in the first line,  changed
the order of the summations in the second line, and
put $ u'= k-u $. Similarly, we have
\begin{equation*}
\begin{split}
B &= \sum_{ k = -2M }^{-1} \sum_{ u= -M }^{M+k} a_1 (k) a_2 (k-u) a_3 (u) \\
&= \sum_{ u = -M }^{M}\sum_{k=u-M}^{-1} a_1 (k) a_2 (k-u) a_3 (u) \\
&= \sum_{ u = -M }^{M} \sum_{u'=-M}^{-u-1} a_1 (u+u') a_2 (u') a_3 (u).
\end{split}
\end{equation*}
Thus we see that
\begin{equation}\label{lm001}
\sum_{ k \in \mathcal{K} } \sum_{ (s,s') \in \mathcal{S}(k) } a_1 (k) a_2 (s) a_3 (s')
=\sum_{ u = -M }^{M} \sum_{u'=-M}^{M} a_1 (u+u') a_2 (u') a_3 (u).
\end{equation}

Applying (\ref{lm001}) when $ a_1 (k) = c (k) e^{2 \pi i k t} $,
$ a_2 (s) = e^{-2 \pi i s t^{j'}_{l'}} $ and $ a_3 (s')
= e^{-2 \pi i s' t^j_l} $, we obtain
\begin{equation}\label{PF}
\begin{split}
& (V_{(\mathcal{K}, \mathcal{S}, c)})_{j,j'} (t)\\
&= \sum_{l=1}^{N_j} \sum_{l'=1}^{N_{j'}}
\sum_{u=-M}^M \sum_{u'=-M}^M c(u-u') e^{2 \pi i u (t-t_l^j)}
e^{-2 \pi i u' (t-t_{l'}^{j'})} \Delta X^j_l
\Delta X^{j'}_{l'} \\
& \hspace{4cm} ( 1 \leq i,j \leq d).
\end{split}
\end{equation}
Here we used an obvious change of variables $ u' \mapsto -u' $.

Now the positive definiteness follows easily. In fact,
for $ x \in \mathbb{C}^d $, we have
\begin{equation*}
\begin{split}
& \sum_{j,j'} (V_{(\mathcal{K}, \mathcal{S}, c)})_{j,j'} (t)x_j \overline{x_{j'}} \\
&=  \sum_{u=-M}^M \sum_{u'=-M}^M c(u-u') \\
& \cdot \left(
\sum_{j=1}^d x_j \sum_{l=1}^{N_j} e^{2 \pi i u (t-t_l^j)}
\Delta X^j_l\right)
\left(
\sum_{j'=1}^d \overline{x_{j'}} \sum_{l=1}^{N_{j'}} e^{-2 \pi i u'
 (t-t_{l'}^{j'})} \Delta X^{j'}_{l'} \right) \\
&= \sum_{u=-M}^M \sum_{u'=-M}^M c(u-u') f (u) \overline{f(u')} \geq 0,
\end{split}
\end{equation*}
where we have put
\begin{equation*}
f (u) := \sum_{j=1}^d x_j \sum_{l=1}^{N_j}
e^{2 \pi i u(t- t_l^j)} \Delta X^{j}_{l}.
\end{equation*}
\end{proof}

\begin{rem}\label{reduction}
If we set in (\ref{PF})
$ N := N_j=N_{j'} = 2M+1 $, $ \Delta t^j_l \equiv 1/N $, and
$ c (k) = 1-\min(|k|, |N-k|)/M $,
the estimator (\ref{PF}) coincides with
(\ref{FEMM09}) with $ L= M $.
In fact,
writing $ t_l = l/N $ for $ l=1, \cdots, N $,
we notice that, for $ t = l_0/N $,
\begin{equation*}
\begin{split}
& (V_{(\mathcal{K}, \mathcal{S}, c)})_{j,j'} (t)\\
&= \frac{1}{N}
\sum_{l,l'} \Delta X^j_l \Delta X^{j'}_{l'}
\sum_{k=-M}^M \sum_{s=-M}^M c(k-s) e^{\frac{2 \pi i k(l_0-l)}{N}}
e^{- \frac{2 \pi i s(l_0-l')}{N} }, \\
\end{split}
\end{equation*}
and by the change of variables $ (k,s) \mapsto (k-s,-s) $,
which is an automorphism over $ \mathbb{Z}/N \mathbb{Z} $,
we have
\begin{equation*}
\begin{split}
&= \frac{1}{N}
\sum_{l,l'} \Delta X^j_l \Delta X^{j'}_{l'}
\sum_{k=-M}^M \sum_{s=-M}^M c(k) e^{\frac{2 \pi i k(l_0-l)}{N}}
e^{\frac{2 \pi i s(l-l')}{N} }.
\end{split}
\end{equation*}
\end{rem}

\subsection{Parameterization by measures}\label{Parametrize}
By Bochner's theorem, we know that for each positive semi-definite
function $ c $,
there exists a bounded measure $ \mu $ on $ \mathbf{R} $ such that
\begin{equation}\label{Boch}
c (x) = \int_\mathbf{R}
e^{2 \pi i y x} \,\mu(dy).
\end{equation}
Therefore, we may rewrite the positive Fourier estimator (\ref{PF})
using the measure $ \mu $
instead of the positive semi-definite function $ c $.
The expression
in terms of the measure $ \mu $ will be useful when implementing the
Fourier method in estimating a spot volatility matrix.  

\begin{defi}\label{factorize}
Let $ \mu $ be a bounded measure and
$ M $ be a positive integer.
We associate with $ (\mu, M) $
an estimator of the spot volatility matrix as:
\begin{equation}\label{PF2}
\begin{split}
& (V_{(\mu,M)})_{j,j'} (t)
= \int_\mathbf{R}
\left( \sum_{l=1}^{N_j}  D_M ( t-t_l^j+y) \Delta X^j_l \right)\\
& \hspace{4cm}
\cdot \left( \sum_{l'=1}^{N_{j'}} D_M (t-t_{l'}^{j'} + y)
\Delta X^{j'}_{l'}\right) \mu(dy),\\
& \hspace{6cm} ( 1 \leq j,j' \leq d).
\end{split}
\end{equation}
\end{defi}

\begin{rem}
Under the assumptions in Theorem \ref{PFE}
with the relation (\ref{Boch}),
we have that $ V_{(\mathcal{K}, \mathcal{S}, c)} (t)
= V_{(\mu, M)} (t) $ for all $ t \in [0,1] $.
In fact, we have
\begin{equation*}\label{PF3}
\begin{split}
& (V_{(\mathcal{K}, \mathcal{S}, c)})_{j,j'} (t) \\
&= \sum_{\substack{1 \leq l \leq N_j \\ 1 \leq l' \leq N_{j'}}}
\int_\mathbf{R} \sum_{ |k| \leq M  } \sum_{|s| \leq M}
e^{2 \pi i (t-t_l^{j}+y)k} e^{-2 \pi i(t-t_{l'}^{j'}+y) s} \mu(dy)
\Delta X^{j}_l \Delta X^{j'}_{l'}\\
&= \sum_{\substack{1 \leq l \leq N_{j} \\ 1 \leq l' \leq N_{j'}}}
\int_\mathbf{R} D_M (t-t^j_l+y) D_M
(t-t^{j'}_{l'}+y) \mu(dy) \Delta X^j_l \Delta X^{j'}_{l'}\\
&=(V_{(\mu,M)})_{j,j'} (t).
\end{split}
\end{equation*}
Note that $ V_{(\mu,M)} $ is easily seen to be real symmetric, and thus
so is $ V_{(\mathcal{K}, \mathcal{S}, c)} $.
Further, it is easier to see that $ V_{(\mu,M)} $ is positive semi-definite.
In fact, for arbitrary $ x = (x_1, \cdots, x_d) \in \mathbf{R}^d $,
\begin{equation*}
\begin{split}
& \sum_{j,j'} (V_{(\mu,M)})_{j,j'} (t) x_j x_{j'} \\
&= \sum_{j,j'}\int_\mathbf{R}
\sum_{\substack{1 \leq l \leq N_j \\ 1 \leq l' \leq N_{j'}}}
D_M (t-t^j_l+y)\Delta X^j_l x_{j}\\
&= \int_\mathbf{R}
\sum_{j=1}^d \sum_{1 \leq l \leq N_j} 
D_M (t-t^j_l+y)\Delta X^j_l x_{j}\\
& \hspace{3cm} \cdot \sum_{j'=1}^d \sum_{1 \leq l' \leq N_{j'}} 
D_M (t-t^{j'}_{l'}+y) \Delta X^{j'}_{l'}
x_{j'} \mu(dy) \\
&= \int_\mathbf{R}
\left(\sum_{j=1}^d  \sum_{1 \leq l \leq N_j} D_M (t-t^j_l+y) x_j \right)^2 \mu(dy) \geq 0.
\end{split}
\end{equation*}
\end{rem}

\subsection{Remarks on the choice of the measure}\label{Measure}
From the observation made in (\ref{heuristic}),
we may insist we choose a sequence of positive semi-definite functions
$ c_N $, where $ N := \max_j N_j $ for simplicity,
so as that
\begin{equation*}
c_N(k) \sim \frac{1}{|\mathcal{S}_N (k)|} C_N (k)
\end{equation*}
where the kernel
\begin{equation}\label{kernel0}
\sum_{k=-2M_N}^{2M_N} C_N (k)e^{2 \pi i k s}
\end{equation}
behaves like/better than Fej\'er one.

The first example is the Fej\'er sum case where
\begin{equation*}
C_N (k) =
1- \frac{|k|}{2M_N+1},
\end{equation*}
or equivalently
\begin{equation*}
c_N (k) = \frac{1}{2M_N +1}
\end{equation*}
and therefore
\begin{equation*}
\mu_N = \frac{1}{2M_N +1} \delta_0.
\end{equation*}
In this case, the convergence of 
$II$ in (\ref{heuristic}) may not be good, 
which might be easier to be seen from the expression of (\ref{PF2}). 

Note that the estimator is completely different from 
the original one (\ref{FEMM09}) since $ |\mathcal{S} (k)| 
= 2M -|k|+1 $ in the former while it is always $ 2M $, 
independent of $ k $, in the latter. 
The factor $ |\mathcal{S} (k)| $ contributes less 
to the consistency in the newly introduced 
positive semi-definite class of estimators.

Looking at the above primitive case, however, 
we notice that a proper choice for the measures 
would be implied by   
\begin{equation*}
(2 M_N +1)^{-1} \times \text{(delta approximating kernel)}.
\end{equation*}

Here we list possible choices.
\begin{ex}\label{Cauchy}
Let 
\begin{equation}\label{Poisson}
C_N (k) = \left(1- \frac{|k|}{2M_N+1}\right)
e^{-\gamma_N |k|},
\end{equation}
where, $ \gamma_N \to 0 $ as $ N \to \infty $.
In this case,
\begin{equation}\label{Cauchy2}
\mu_N (dy) = \frac{1}{2M_N +1} \frac{1}{\pi}
\frac{\gamma_N}{y^2 + \gamma_N^2} dy,
\end{equation}
a Cauchy kernel.
\end{ex}

\begin{ex}\label{Gaussian}
Let 
\begin{equation}\label{Gauss}
C_N (k) =  \left(1- \frac{|k|}{2M_N+1}\right)
e^{-\frac{2 \pi^2 k^2}{L_N}},
\end{equation}
where, $ L_N \to \infty $ as $ N \to \infty $.
In this case,
\begin{equation}\label{Gauss2}
\mu_N (dy) = \frac{1}{2 M_N +1} \sqrt{\frac{L_N}{2\pi}}  e^{-L_N y^2} dy,
\end{equation}
a Gaussian kernel.
\end{ex}

\begin{ex}\label{doubleF}
We let
\begin{equation}\label{F1}
C_N (k) =  \left(1- \frac{|k|}{2M_N+1}\right)^2.
\end{equation}
In this case, its corresponding measure is the Fej\'er kernel;
\begin{equation}\label{F2}
\mu_N (\{y\}) = \frac{1}{2M_N+1} \left( \frac{\sin (2M_N+1) \pi y}
{\sin \pi y}
\right)^2 = K_{2M_N+1} (y),
\end{equation}
$ y = \frac{k}{2M_N+1}, k= 0,1,\cdots, 2M_N $ if $ 2M_N +1 $
is a prime number.
This can be seen by the following relation:
\begin{equation*}
1 - \frac{|k|}{L}
= \sum_{t= 0}^{L-1} K_{L} (t) e^{-2 \pi i k t},
\end{equation*}
which is valid when $ L $ is a prime number and is implied by
\begin{equation*}
K_{L} (t) = \sum_{k=-(L-1)}^{L-1} \left(1 - \frac{|k|}{L}
\right) e^{2 \pi i k t}.
\end{equation*}
It is notable that in this case we need not
discretize the integral with respect to $ \mu_N $
since it is already discrete.
\end{ex}

In the use of a delta kernel,
one needs to choose properly the approximating parameters
of the kernel as well as $ M_N $; the delta approximating parameters
are $ \gamma_N $ in Example \ref{Cauchy} and
$ L_N $ in Example \ref{Gaussian}.
(The Fej\'er case of Example \ref{doubleF} is an exception).
Even with a consistency
result which only tells an asymptotic behavior,
one still needs to optimize the choice with some other criteria.
In the next section, we give some simulation results
to have a clear view of this issue.

\section{Experimental Results}\label{experiments}
In this section, we present some results
of simple experiments to exemplify how our method is
implemented;
\begin{itemize}
\item We applied our estimation method to
the daily data from \\
31/03/2008 to 26/09/2008 
of zero-rate implied by Japanese government bond prices with maturities
07/12 and 07/06, from 07/12/2008 to 07/06/2014.
\item Therefore, we set $ N = 150 $ ($ = N_j $ for all $ j $,
the observation dates are equally spaced)
and $ d= 12 $.
\item We set $ M = M_N = 15 $ for $ M $ in (\ref{PF2})
and $ M_N $ in (\ref{kernel0}).
\item The integral with respect to $ \mu $ is  also discretized;
we only use $ [-1/2, 1/2] $ instead of whole real line,
which is discretized to $ \{ -1/2 + k/(2M_N+1); k= 0, 1, \cdots, 2 M_N \} $.
\item We tested the Cauchy kernel estimator with (\ref{Poisson}) in 
Example \ref{Cauchy}
in Experiment 1 and 2 with different $ \gamma_N $,
the Gaussian kernel ones of Example \ref{Gaussian}
in Experiment 3 and 4 with different $ L_N $.
\item We used Octave ver. 3.2.4, 
and a Vaio/SONY, Windows 7 64bit OS laptop PC, with processor
Intel(R) Core(TM) i3-2310M CPU @2.10GHz 2.10GHz,
and RAM 4.00 GB.
\end{itemize}
All the figures are indicating the results of
``dynamical principal component analysis",
where the graphs from the top shows
the time evolution of the rate of
the biggest, the biggest + the second, and
the first three biggest eigenvalues, respectively.
Each experiment took about 3 minutes; plausibly fast.
We see the similarities between Figure \ref{w1} and 
Figure \ref{w2}, and between Figure \ref{w4} and Figure \ref{w5}. 
In these experiments, we should say that 
the accuracy is not fully appreciated, but 
we may say that the order of the delta kernel is important to 
have an accuracy. 

\begin{figure}
\caption{Experiment 1; $ \gamma_N = (2M_N+1)^{-1/2} \approx 0.1796 $}
\label{w1}
\includegraphics[width=12cm]{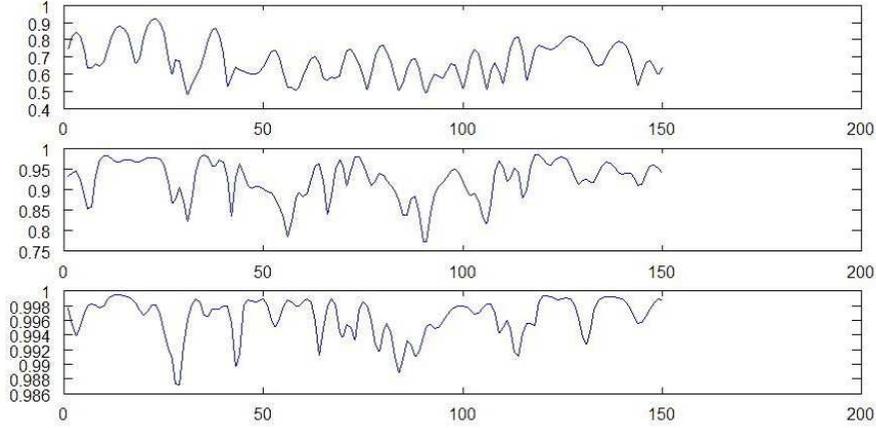}
\end{figure}
\begin{figure}
\caption{Experiment 2;$ \gamma_N = (2M_N+1)^{-1/4} \approx 0.4238 $}
\label{w2}
\includegraphics[width=12cm]{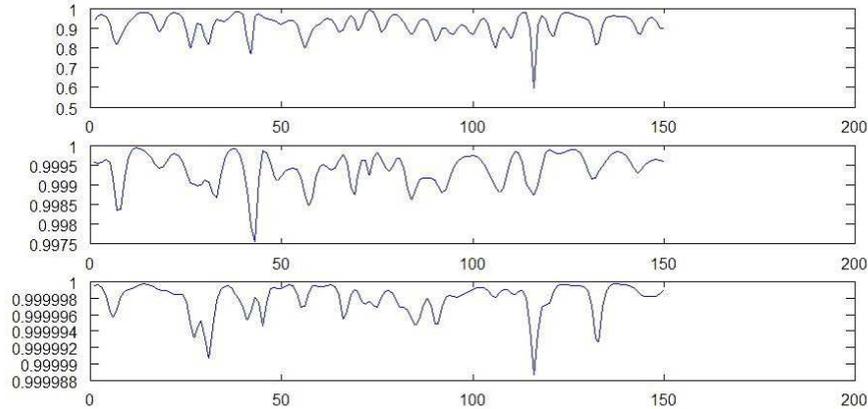}
\end{figure}

\if1
\begin{figure}
\caption{Experiment 3;$ \gamma_N = N^{-9/8} \approx 0.0036 $}
\label{w3}
\includegraphics[width=12cm]{Pattern7.eps}
\end{figure}
\fi


\begin{figure}
\caption{Experiment 3; $ L_N =2M_N+1 =31 $}
\label{w4}
\includegraphics[width=12cm]{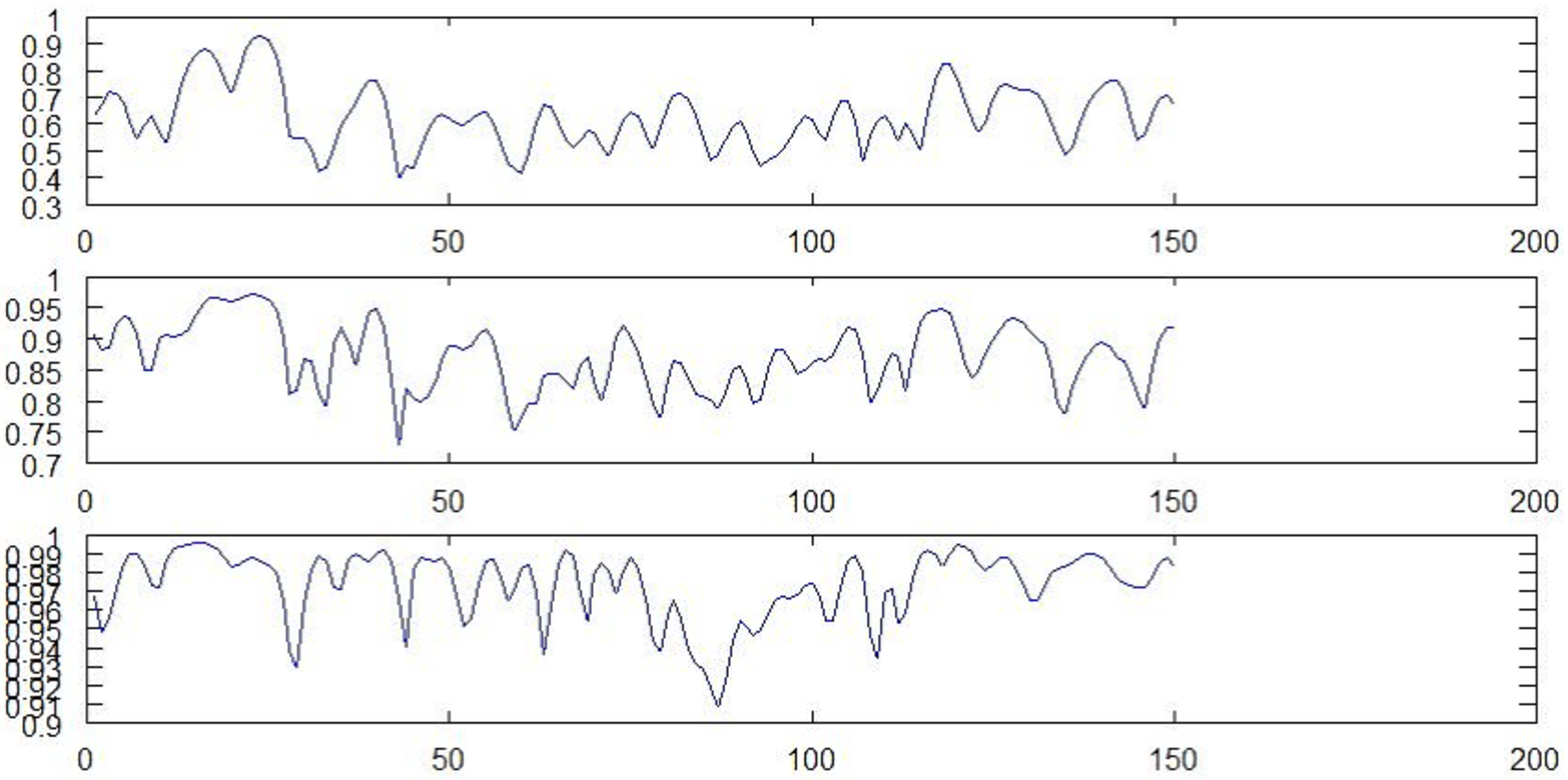}
\end{figure}

\begin{figure}
\caption{Experiment 4;$L_N = (M_N+1)^{1/4}  \approx 2.36 $}
\label{w5}
\includegraphics[width=12cm]{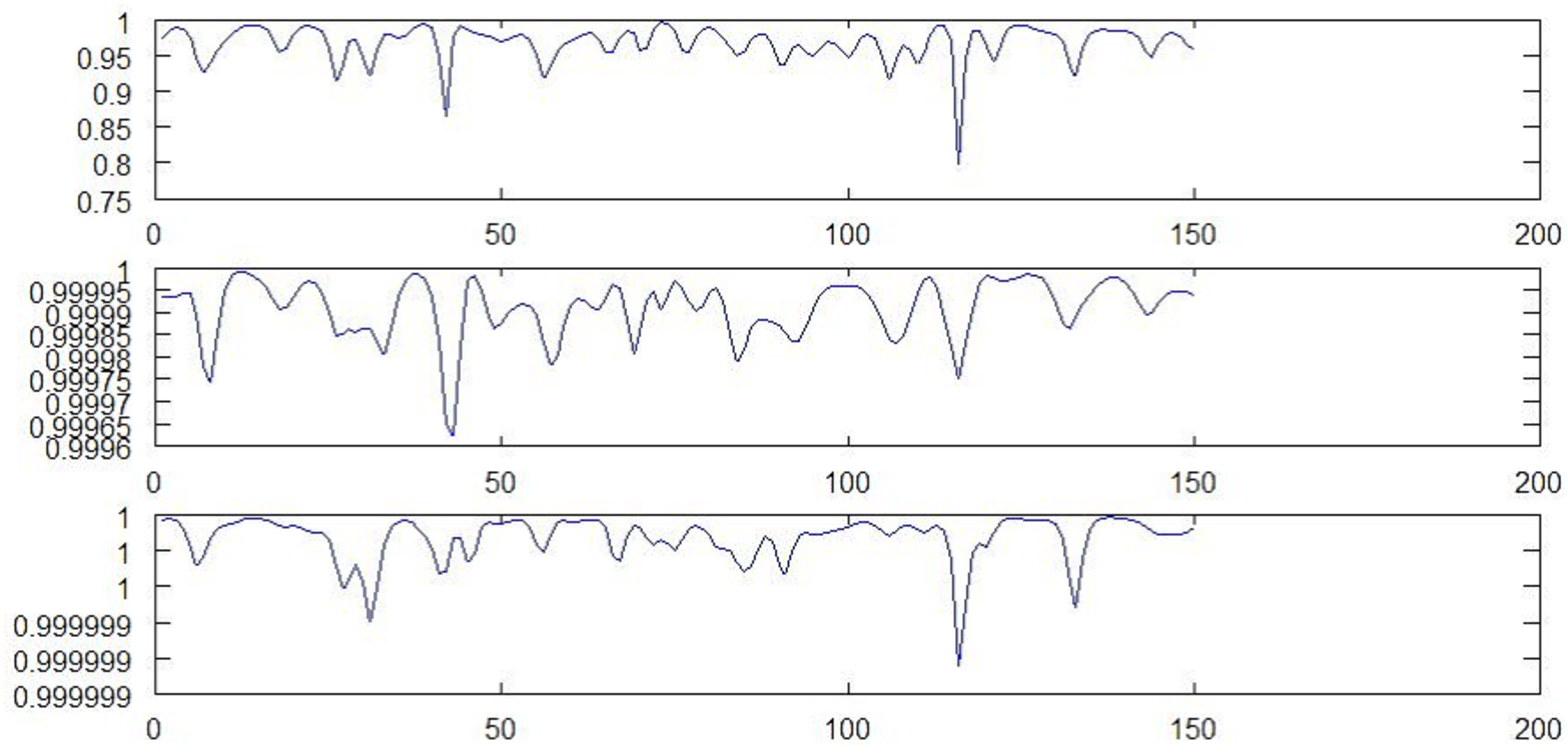}
\end{figure}


\if2
\begin{figure}
\caption{Experiment 4;$L_N =N^{33/16} \approx 3.0 \times10^4 $}
\label{w6}
\includegraphics[width=12cm]{Pattern8.eps}
\end{figure}
\fi


\begin{thebibliography}{99}


\bibitem{LM12}
Liu, N.L., and Mancino, M.E. (2012)
``Fourier estimation method applied to forward interest rates",
{\em JSIAM Letters} 4, 17-20.

\bibitem{LN12}
Liu, N.L., and Ngo, H.L. (2014)
``Approximation of eigenvalues of spot cross volatility matrix with a view toward principal component analysis", arXiv:1409.2214 [q-fin.ST]. 

\bibitem{MM02}
Malliavin, P. and Mancino, M.E. (2002)
``Fourier Series Method for Measurement of Multivariate Volatilities",
{\em Finance and Stochastics} , {\bf 6}, pp.49--61.

\bibitem{MM09}
Malliavin, P. and Mancino, M.E. (2009)
Fourier Transform Method for Nonparametric Estimation of Multivariate
Volatility {\it The Annals of Statistics},
vol. 37, No. 4, pp.1983--2010.

\bibitem{MS11}
Mancino, M.E. and Sanfelici, S.:
Estimating Covariance via Fourier Method
in the Presence of Asynchronous Trading and Microstructure Noise
{\it Journal of Financial Econometrics} (2011),
vol. 9, No. 2, pp.367--408.


\bibitem{NO09}
Ngo, H.L. and Ogawa, S.:
A Central Limit Theorem for
the Functional Estimation of the Spot Volatility
{\it Monte Carlo Methods and Applications} (2009),
vol. 15, No. 4, pp.353--380.

\end{thebibliography}
\end{document}